\documentclass[smallabstract,smallcaptions]{dccpaper}
\usepackage{epsfig}
\usepackage{cite}
\usepackage{amsmath}
\usepackage{amssymb}
\usepackage{color}
\usepackage{url}

\usepackage{amsthm}

\newlength{\figurewidth}
\newlength{\smallfigurewidth}

\newtheorem{proposition}{Proposition}

\begin{document}
\title{\large\textbf{Exploiting Computation-Friendly\\
Graph Compression Methods for\\
Adjacency-Matrix Multiplication}}
\author{%
Alexandre P Francisco$^{\ast}$, Travis Gagie$^{\dag}$, Susana Ladra$^{\ddagger}$, and Gonzalo Navarro$^{\mathsection}$\\[0.5em]
{\small\begin{minipage}{\linewidth}\begin{center}
\begin{tabular}{ccc}
$^{\ast}$INESC-ID / IST & \hspace*{0.5in} & $^{\dag}$EIT, Diego Portales University \\
Universidade de Lisboa && and CeBiB\\
Portugal && Chile\\
\url{aplf@ist.utl.pt} && \url{travis.gagie@gmail.com} \\
         &&      \\
$^{\ddagger}$Facultade de Inform\'{a}tica / CITIC & \hspace*{0.5in} & $^{\mathsection}$Department of Computer Science \\
Universidade da Coru\~{n}a && University of Chile\\
Spain && Chile\\
\url{sladra@udc.es} && \url{gnavarro@dcc.uchile.cl}
\end{tabular}
\end{center}\end{minipage}
}}


\maketitle              
\thispagestyle{empty}

\begin{abstract}
Computing the product of the (binary) adjacency matrix of a large graph with 
a real-valued vector is an important operation that lies at the heart of various
graph analysis tasks, such as computing PageRank. In this paper we show that
some well-known Web and social graph compression formats are {\em 
computation-friendly}, in the sense that they allow boosting the computation.
In particular, we show that the format of Boldi and Vigna allows computing the
product in time proportional to the compressed graph size. Our experimental
results show speedups of at least 2 on graphs that were compressed at least
5 times with respect to the original. We show that other successful graph
compression formats enjoy this property as well.
\end{abstract}
\Section{Introduction} \label{sec:introduction}
Let $\mathbf{A}$ be an $n\times n$ binary matrix and $\mathbf{x}\in I\!\!R^n$ a vector.
Matrix vector multiplication, either $\mathbf{x}\cdot\mathbf{A}$ or $\mathbf{A}\cdot\mathbf{x}^\top$, is not only a fundamental operation in mathematics, but
also a key operation in various graph-analysis tasks, when $\mathbf{A}$ is their
adjacency matrix. A well-known example, which we use as a motivation, is the 
computation of PageRank on large Web graphs. PageRank is a particular case of 
many network centrality measures that can be approximated through the power 
method~\cite{new10}. Most real networks, and in particular Web and social graphs,
have very sparse adjacency matrices~\cite{chung06}.
While it is straightforward to compute a matrix-vector product in time 
proportional to the nonzero entries of $\mathbf{A}$, the most successful Web
and social graph compression methods exploit other properties that allow them to
compress the graphs well beyond what is possible by their mere sparsity.
It is therefore natural to ask whether those more powerful compression formats 
allow us, as sparsity does, to compute the product {\em in time proportional to
the size of the compressed representation}. This is an instance of {\em
computation-friendly compression}, which seeks for compression formats that
not only reduce the size of the representation of objects, but also speeds up
computations on them by directly operating on the compressed representations.
Other examples of computation-friendly compression are pattern matching in compressed strings~\cite{GGP15}, computation of edit distance between compressible strings~\cite{HLLW13}, speedups for multiplying sequences of matrices and the Viterbi algorithm~\cite{LMWZ09}, building small and shallow circuits~\cite{GHJLN17}, among other tasks~\cite{Loh12}.

In this paper we exploit compressed representations of Web and social networks
and show that matrix-vector products can be carried out much faster than just
operating on all the nonzero entries of the matrix.
Although our approach can be extended to other compressed representations of
graphs and binary matrices, we mostly consider the representation proposed by 
Boldi and Vigna~\cite{bv04}. The relevant observation for us is that adjacency 
lists, i.e., rows in $\mathbf{A}$, are compressed differentially with respect
to other similar lists, and thus one can reuse and ``correct'' the result of 
the multiplication of a previous similar row with $\mathbf{x}^\top$.

We describe previous work in the next section. The following sections describe PageRank and the compression format of Boldi and Vigna. We then describe how we exploit that compression format to speed up matrix multiplication.
The following section contains experimental results, and we conclude
with a discussion of other compression formats that
favor matrix multiplications, and future work directions.

\Section{Previous Work}

Matrix multiplication is a fundamental problem in computer science; see, e.g.,~\cite{AVW18} for a recent survey of results.
Computation-friendly matrix compression has been already considered by others, even if indirectly.
Karande {et al.}~\cite{doi:10.1080/15427951.2009.10390646} addressed it by
exploiting a structural compression scheme, namely by introducing virtual
nodes. Although their results were similar to the ones presented in this
paper, their approach was more complex and it could not be used directly,
requiring the correction of computation results. On the other hand, contrary to
their belief, we show in this paper that representational compression schemes
do not always require the same amount of computation, providing a much simpler
approach that can be used directly without requiring corrections.

Another interesting approach was proposed by Nishino {et al.}~\cite{doi:10.1137/1.9781611973440.122}.
Although they did not exploit compression in the same way we do, they observed that intermediate
computational results for the matrix multiplication of equivalent partial
rows of a matrix are the same. They used then an adjacency forest where
rows are represented by sharing common suffixes. We should note that
the authors consider general real matrices, and not only Boolean matrices as we
do. Nevertheless they presented results for computing the PageRank over
adjacency matrices as we do, achieving similar results. Their approach
implied preprocessing the graph, however, while we start from an already
compressed graph. An interesting question is how their approach could 
be exploited on top of $k^2$-trees~\cite{Brisaboa2014}.

The question addressed here can also be of interest for the 
problem of Online Matrix-Vector (OMV) multiplication. Given
a stream of binary vectors, $\mathbf{x}_1,\mathbf{x}_2,\mathbf{x}_3,\ldots$, the results of matrix-vector multiplications
$\mathbf{x}_i\cdot\mathbf{A}$ can be computed faster than computing them independently, with
most approaches making use of previous computations $\mathbf{x}_j\cdot\mathbf{A}$, for $j<i$, to speed up
the computation of each new product $\mathbf{x}_i\cdot\mathbf{A}$~\cite{doi:10.1145/2746539.2746609,arxiv:1605.01695}.
Nevertheless, none of those approaches preprocess matrix $\mathbf{A}$
to exploit its redundancies. Hence, by exploiting a suitable succinct
representation of $\mathbf{A}$ as we do here, an improvement for OMV can
be easily obtained, with computational time depending on the
length of the succinct representation of $\mathbf{A}$ instead.

\Section{PageRank} \label{sec:pagerank}

Given $G=(V,E)$ a graph with $n = |V|$ vertices and $m = |E|$ edges, let $\mathbf{A}$ be its adjacency matrix; $A_{u v} = 1$ if $(u, v)\in E$, and $A_{u v} = 0$ otherwise.
The normalized adjacency matrix of $G$ is the matrix $\mathbf{M} = \mathbf{D}^{-1}\cdot\mathbf{A}$, where $\mathbf{D}$ is an $n\times n$ diagonal matrix with $D_{u u}$ the degree $d_u$ of $u\in V$, i.e., $D_{u u} = d_u = \sum_v A_{u v}$.
Note that $\mathbf{M}$ is the standard random walk matrix, where a random walker at vertex $u$ jumps to a neighbor $v$ of $u$ with probability $1/d_u$.
Moreover the $k$-power of $\mathbf{M}$, $\mathbf{M}^k$, is the random walk matrix after $k$ steps, i.e., $M_{u v}^k$ is the probability of the random walker being at vertex $v$ after $k$ jumps, having started at vertex $u$.
PageRank is a typical random walk on $G$ with transition matrix $\mathbf{M}$.
Given a constant $0 <\alpha < 1$ and a probability vector $p_0$, the PageRank vector $\mathbf{p}_\alpha$ is given by the following recurrence~\cite{chung07}:
\begin{equation*}
\mathbf{p}_\alpha = \alpha \mathbf{p}_0 + (1-\alpha) \mathbf{p}_\alpha\cdot\mathbf{M}\,.
\end{equation*}
The parameter $\alpha$ is called the teleport probability or jumping factor, and $\mathbf{p}_0$ is the starting vector.
In the original PageRank~\cite{pr99}, the starting vector $\mathbf{p}_0$ is the uniform distribution over the vertices of $G$, i.e., $\mathbf{p}_0 = \mathbf{1}/n$.
When $\mathbf{p}_0$ is not the stationary distribution, $\mathbf{p}_\alpha$ is called a personalized PageRank. Intuitively, $\mathbf{p}_\alpha$ is the probability of a lazy Web visitor to be at each page assuming that he/she surfs the Web by either randomly starting at a new page or jumping through a link from the current page.
The parameter $\alpha$ ensures that such a surfer does not get stuck at a dead end.
PageRank can be approximated iteratively through the power iteration method by iterating, for $t\geq 1$:
\begin{equation}\label{eq:priter}
\mathbf{p}_t = \alpha \mathbf{p}_0 + (1-\alpha) \mathbf{p}_{t - 1}\cdot\mathbf{M}\,.
\end{equation}

We show how to speed up these matrix-vector multiplications when the adjacency
matrix $\mathbf{A}$ is compressible.

\Section{Our Approach} \label{sec:approach}
Our main idea is to exploit the copy-property of adjacency lists observed in some graphs, such as Web graphs~\cite{bv04}.
The adjacency lists of neighbor vertices tend to be very similar and, hence, the rows in the adjacency matrix are also very similar. Moreover these networks reveal also strong clustering effects, with local groups of vertices being strongly connected and/or sharing many neighbors.
The copy-property effect can then be further amplified through clustering and suitable vertex reordering, an important step for achieving better graph compression ratios~\cite{BRSLLP}.
Most compressed representations for sparse graphs rely on these properties~\cite{BLNis13.2,GB11,CNtweb10}.
In this paper, we consider the WebGraph framework, a suite of codes, algorithms and tools that aims at making it easy to manipulate large Web graphs~\cite{bv04}.
Among several compression techniques used in Webgraph, our approach makes use of list referencing.

Let $\mathbf{A}$ be an $n\times n$ binary sparse matrix,
\begin{equation*}
\mathbf{A} = \begin{bmatrix}
\mathbf{v}_1\\
\vdots\\
\mathbf{v}_n
\end{bmatrix}
\end{equation*}
where $\mathbf{v}_i\in\{0,1\}^n$ is the $i$-th row, for $i = 1, \ldots, n$.
Let $\mathbf{r}\in\{0, 1, \ldots, n\}^n$ be a referencing vector such that, for $i\in \{1, \ldots, n\}$, $r_i < i$ and $\mathbf{v}_{r_i}$ is some previous row used for representing $\mathbf{v}_i$. 
Let also $\mathbf{v}_0 = \mathbf{0}$ and $r_1 = 0$.
The reference $r_i$ is found in the WebGraph framework within a given window $W$, i.e., $r_i\in\{\max(1,i-W), \dots,i\}$, and it is optimized to reduce the length of the representation of $\mathbf{v}_i$.
The line $\mathbf{v}_i$ is then represented by adding missing entries and marking spurious ones, with respect to $\mathbf{v}_{r_i}$, and encoded using several techniques, such as differential compression and codes for natural numbers~\cite{bv04,bv05}.

\begin{proposition}
Given an $n\times n$ matrix $\mathbf{A}$, $\mathbf{x}\in I\!\!R^n$, and a referencing vector $\mathbf{r}$ for $\mathbf{A}$, let $\mathbf{A}'$ and $\mathbf{w}$ be defined as follows:
\begin{equation*}
\mathbf{A}' = \begin{bmatrix}
\mathbf{v}_1 - \mathbf{v}_{r_1}\\
\vdots\\
\mathbf{v}_n - \mathbf{v}_{r_n}
\end{bmatrix}
\end{equation*}
\begin{equation*}
w_i = \mathbf{v}_{r_i} \cdot \mathbf{x}^\top
\end{equation*}
Then we have that:
\begin{equation*}
\mathbf{A} \cdot \mathbf{x}^\top = \mathbf{A}' \cdot \mathbf{x}^\top + \mathbf{w}^\top
\end{equation*}
\end{proposition}
\begin{proof}
By definition,
\begin{equation*}
\mathbf{A}'\cdot \mathbf{x}^\top + \mathbf{w}^\top =
\begin{bmatrix}
\mathbf{v}_1\cdot \mathbf{x}^\top - \mathbf{v}_{r_1}\cdot \mathbf{x}^\top \\
\vdots\\
\mathbf{v}_n\cdot \mathbf{x}^\top - \mathbf{v}_{r_n}\cdot \mathbf{x}^\top
\end{bmatrix}
+ 
\begin{bmatrix}
\mathbf{v}_{r_1} \cdot \mathbf{x}^\top\\
\vdots\\
\mathbf{v}_{r_n} \cdot \mathbf{x}^\top
\end{bmatrix}
= 
\begin{bmatrix}
\mathbf{v}_1\cdot \mathbf{x}^\top  \\
\vdots\\
\mathbf{v}_n\cdot \mathbf{x}^\top
\end{bmatrix}
= \mathbf{A} \cdot \mathbf{x}^\top
\end{equation*}
\end{proof}

Let us compute $\mathbf{y}^\top = \mathbf{A} \cdot \mathbf{x}^\top$ by iterating over $i=1,\ldots,n$. Then $\mathbf{w}$ can be incrementally computed because $r_i < i$ and $w_i = y_{r_i}$, ensuring that $w_{i}$ is already computed when required to compute $y_i$.
Given inputs $\mathbf{A}'$, $\mathbf{r}$ and $\mathbf{x}$, the algorithm to compute $\mathbf{y}$ is as follows:
\begin{enumerate}
\item Set $\mathbf{y} = \mathbf{0}$ and $y_0 = 0$.
\item For $i = 1,\ldots,n$, set $y_i = y_{r_i} + \sum_{j} A'_{i j} x_j$.
\item Return $\mathbf{y}$.
\end{enumerate}
Note that the number of operations to obtain $\mathbf{y}^\top = \mathbf{A} \cdot \mathbf{x}^\top$ is proportional to the number of nonzeros in $\mathbf{A}'$, that is, to the
compressed representation size.
Depending on the properties of $\mathbf{A}$ discussed before, this number may be much smaller than the number of nonzeros in $\mathbf{A}$.
We present in the next section experimental results for Web graphs, where we indeed obtain considerable speedups the computation of PageRank.

\Section{Experimental Evaluation} \label{sec:exp}
We computed the number of nonzeros $m'$ in $\mathbf{A'}$ for the adjacency matrix $\mathbf{A}$ of several graphs available at \texttt{http://law.di.unimi.it/datasets.php}~\cite{bv04,BRSLLP,BMSB}.
Whenever $|\mathbf{v}_i - \mathbf{v}_{r_i}| \geq |\mathbf{v}_i|$, we kept $\mathbf{v}_i$ as the row in $\mathbf{A'}$, since it resulted in fewer nonzeros.
Results are presented in Table~\ref{tab:datasets}, including the number of vertices $n$ and the number of edges $m$, for each graph.
Both $\mathbf{A'}$ and $\mathbf{r}$ were obtained directly from the Webgraph representation using high compression, which uses stronger referencing among adjacencies and thus favors our approach.
\begin{table}[!t]
\centering
\caption{Datasets used in the experimental evaluation, where $n$ is the number of vertices, $m$ is the number of edges (i.e., nonzeros in $\mathbf{A}$), $m'$ is the number of nonzeros in $\mathbf{A'}$, $t$ is the average time in seconds to compute a matrix-vector product with $\mathbf{A}$, $t'$ is the average time in seconds to compute a matrix-vector product with $\mathbf{A'}$, and $S$ is the speedup observed in the computation of PageRank. The first five datasets are Web crawls and the remaining ones are social networks. All datasets are available at \texttt{http://law.di.unimi.it/datasets.php}. Times and speedups were only computed for web graphs.}
\label{tab:datasets}
\vskip1ex
{\small
\begin{tabular}{|l|r|l|l|c|c|c|c|}
\hline
Graph                  & \multicolumn{1}{|c|}{$n$} & \multicolumn{1}{|c|}{$m$} & \multicolumn{1}{|c|}{$m'$} & \multicolumn{1}{|c|}{$m/m'$} & \multicolumn{1}{|c|}{$t$} & \multicolumn{1}{|c|}{$t'$} & \multicolumn{1}{|c|}{$S$} \\ \hline
eu-2015-hc              & $1.07\!\times\!10^9$ & $9.17\!\times\!10^{10}$ & $1.11\!\times\!10^{10}$ & 8.26 & 3244.0 & 1099.0 & 2.95 \\           
eu-2015-host-hc         & $1.13\!\times\!10^7$ & $3.87\!\times\!10^8$    & $1.10\!\times\!10^8$    & 3.52 & 11.55  & 8.15   & 1.42 \\           
gsh-2015-hc             & $9.88\!\times\!10^8$ & $3.39\!\times\!10^{10}$ & $7.08\!\times\!10^9$    & 4.78 & 1803.6 & 953.4  & 1.89 \\           
it-2004-hc              & $4.13\!\times\!10^7$ & $1.15\!\times\!10^9$    & $2.27\!\times\!10^8$    & 5.08 & 24.65  & 12.0   & 2.05 \\           
uk-2014-hc              & $7.88\!\times\!10^8$ & $4.76\!\times\!10^{10}$ & $6.26\!\times\!10^9$    & 7.58 & 2034.0 & 665.8  & 3.05 \\ \hline    
twitter-2010-hc         & $4.17\!\times\!10^7$ & $1.47\!\times\!10^9$    & $1.44\!\times\!10^9$    & 1.02 &   --   &   --   &   -- \\
amazon-2008-hc          & $7.35\!\times\!10^5$ & $5,16\!\times\!10^6$    & $4.48\!\times\!10^6$    & 1.15 &   --   &   --   &   -- \\
enwiki-2013             & $4.21\!\times\!10^6$ & $1.01\!\times\!10^8$    & $9.62\!\times\!10^7$    & 1.05 &   --   &   --   &   -- \\
wordassoc.-2011    & $1.06\!\times\!10^4$ & $7.22\!\times\!10^4$    & $7.15\!\times\!10^4$    & 1.01 &   --   &   --   &   -- \\ \hline
\end{tabular}
}
\end{table}

As expected, our approach works extremely well for Web graphs, with the number of nonzeros in $\mathbf{A}'$ being less than 20\% for page graphs and less than 30\% for host graphs.
Note that Web graphs are known to verify the copy-property among adjacencies.
Other networks we tested, instead, seem not to verify this property in the same degree, and therefore our approach is not beneficial.
This was expected, as social networks are not as compressible as Web graphs~\cite{Chierichettikdd09}.
There may exist, however, other representations for these networks that may 
benefit from other compression approaches (see the next section).

We implemented PageRank using the algorithm above to compute matrix vector products.
Since Eq.~(\ref{eq:priter}) uses left products and our representation is row-oriented, we use the transposed adjacency matrix and right products.
The implementation is in Java and based on the Webgraph representation, where $\mathbf{A}'$ is represented as two graphs: a positive one for edges with weight $1$, and a negative one for edges with weight $-1$.
All tests were conducted on a machine running Linux, with an Intel(R) Xeon(R) CPU E5-2630 v3 @ 2.40GHz (8 cores, cache 32KB/4096KB) and with 32GB of RAM.
Java code was compiled and executed with OpenJDK 1.8.0\_131.

We ran 10 iterations for the Web graphs in Table~\ref{tab:datasets}, starting with the uniform distribution.
Let us consider the graphs \texttt{eu-2015-host-hc} and \texttt{it-2004-hc}.
%
Our implementation took 81.5 and 120.0 seconds for \texttt{eu-2015-host-hc} and \texttt{it-2004-hc}, respectively.
An equivalent implementation of PageRank, using the adjacency matrix $\mathbf{A}$ instead of $\mathbf{A'}$, represented with WebGraph, took 115.5 seconds and 246.5 for \texttt{eu-2015-host-hc} and \texttt{it-2004-hc}, respectively.
Hence, we achieved speedups of 1.42 and 2.05, respectively, as presented in Table~\ref{tab:datasets}.
Observed speedups are lower than what we would expect given that $\mathbf{A}'$ has 3.52 times fewer nonzeros than $\mathbf{A}$ for \texttt{eu-2015-host-hc}, and 5.08 times fewer for \texttt{it-2004-hc}.
After profiling we could observe that, although $\mathbf{A}'$ had much fewer nonzeros than $\mathbf{A}$, the nonzeros in $\mathbf{A}'$ are more dispersed than those in $\mathbf{A}$, with $\mathbf{A}$ benefiting from  contiguous memory accesses.
The speedups are nevertheless significant, namely when we are dealing with larger graphs like \texttt{eu-2015-hc}. Our implementation took 1h30m for this graph, about 3 times less than the equivalent implementation using matrix $\mathbf{A}$ instead of matrix $\mathbf{A'}$.

We replicated the experiments with code written in \texttt{C} using a plain representation
for sparse matrices, for both $\mathbf{A}$ and $\mathbf{A}'$.
The operations became 10 times faster, but the difference between operating with
both $\mathbf{A}$ and $\mathbf{A}'$ remained similar.

%
%



\Section{Final Remarks} \label{sec:final}
We have shown that the adjacency matrix compression scheme of Boldi and Vigna~\cite{bv04}
allows for computing matrix-vector products in time proportional to the 
{\em compressed} matrix size. Therefore, compression not only saves space but 
also speeds up an operation that is key for graph analysis tasks.

This is not a property unique to that compression format. Another suitable
format is the biclique extraction method of Hern\'andez and
Navarro~\cite{HNkais13}. They decompose the edges of $G$ into a number of
bicliques $(S_r,C_r)$, so that every node from $S_r$ points to every node from
$C_r$, plus a residual set of edges. The $|S_r| \cdot |C_r|$ edges of each
biclique are represented in $|S_r| + |C_r|$ words, by just listing both sets.
This format is shown to be competitive to compress both Web and social graphs.
In order to compute $\mathbf{A} \cdot \mathbf{x}^\top$, we compute for each 
biclique $r$ the value $c_r = \sum_{j \in C_r} x_j$. We then initialize $n$ 
counters $y_j = 0$ and, for each biclique $r$ and each $i \in S_r$, we add 
$c_r$ to $y_i$. Finally,
for each residual edge $A_{ij}=1$, we add $x_j$ to $y_i$. The final answer is 
the vector $\mathbf{y}^\top$, which is obtained in time proportional to the
size of the compressed matrix.

We plan to study the practical speedup obtained with this compression format. We also plan to improve the results on Boldi and Vigna's algorithm by varying the size of the window and splitting the input matrix into submatrices of consecutive columns so matches are more flexible and need not span entire rows. We will also consider
other formats where it is less clear how to translate the reduction
in space into a reduction in computation time~\cite{BLNis13.2,GB11,CNtweb10,HNkais13}, and 
study which other relevant matrix operations can be boosted
by which compression formats.

\Section{Acknowledgments}
This research has received funding from the European Union's Horizon
2020 research and innovation programme under the Marie Sk{\l}odowska-Curie
[grant agreement No 690941], namely while the first author was
visiting the University of Chile, and while the second author was
affiliated with the University of Helsinki and visiting the University
of A Coru\~na.
The first author was funded by Funda\c{c}\~{a}o para a Ci\^{e}ncia e a
Tecnologia (FCT) [grant number UID/CEC/50021/2013];
the second author was funded by Academy of Finland
[grant number 268324] and Fondecyt [grant number 1171058]; the third
author was funded by Ministerio de  Econom\'{\i}a y Competitividad
(PGE and FEDER) [grant number TIN2016-77158-C4-3-R] and Xunta de
Galicia (co-founded with FEDER) [grant numbers ED431C 2017/58;
ED431G/01]; and the fourth author was funded by Millennium Nucleus
Information and Coordination in Networks [grant number ICM/FIC
RC130003].

%

%
%

\end{document}